\documentclass{llncs}

\usepackage{amssymb}
\usepackage{textcomp}
\usepackage{comment}
\usepackage{color}
\usepackage{graphicx}
\usepackage{amsmath}

\newcommand{\Beginproof}{{\em Proof.}  }
\newcommand{\Endproof}{\hfill$\Box$}

\DeclareMathOperator{\exv}{\mathbb{E}\textrm{ }}

\begin{document}

\title{Quantum versus Classical Online Streaming Algorithms with Advice}
\author{Kamil~Khadiev$^{1,2,3}$ \and Aliya~Khadieva$^{1,2}$ \and Mansur Ziatdinov$^{2}$ \and Dmitry Kravchenko$^{1}$ \and Alexander Rivosh$^{1}$ \and Ramis Yamilov$^{4}$ \and Ilnaz Mannapov$^{2,3}$}

\institute{
University of Latvia, Riga, Latvia
\and
Kazan Federal University, Kazan, Russia
  \and
Smart Quantum Technologies Ltd., Kazan, Russia
\and Yandex Ltd., Saint Petersburg, Russia 
                \\ \email{kamilhadi@gmail.com, aliya.khadi@gmail.com, gltronred@gmail.com, kdmitry@gmail.com, alexander.rivosh@lu.lv, ramis.yamilov@gmail.com, ilnaztatar5@gmail.com}
}

\maketitle

\begin{abstract}We consider online algorithms with respect to the competitive ratio. Here, we investigate quantum and classical one-way automata with non-constant size of memory (streaming algorithms) as a model for online algorithms. We construct problems that can be solved by quantum online streaming algorithms better than by classical ones in a case of logarithmic or sublogarithmic size of memory, even if classical online algorithms get advice bits. 
Furthermore, we show that a quantum online algorithm with a constant number of qubits can be better than any deterministic online algorithm with a constant number of advice bits and unlimited computational power.
 \\
\textbf{Keywords:} quantum computation, online algorithm, automaton, streaming algorithm, online minimization problem, branching program, BDD
\end{abstract}

\section{Introduction}

Online algorithms are a well-known computational model for solving optimization problems.
The peculiarity is that the algorithm reads an input piece by piece
and should return an answer piece by piece immediately, even if an answer can
depend on future pieces of the input. The algorithm should return an answer for minimizing an objective function (the cost of an output). The most standard method to define the effectiveness is the competitive ratio \cite{st85,kmrs86}. 
%
Typically, online algorithms have unlimited computational power.
The main restriction is a lack of knowledge on future input variables.
At the same time, it is interesting to solve online minimization problems in a case of a big input stream such that the stream cannot be stored completely in the memory. In that case, we can consider automata (streaming algorithm) as online algorithms. This model was explored in \cite{bk2009,gk2015,blm2015,kkm2018}. We are interested in {\em quantum online algorithms}. This model was introduced in \cite{kkm2018} and discussed in \cite{aakv2018}.
It is known that quantum online algorithms can be better than classical ones in the case of sublogarithmic size of memory \cite{kkm2018}. Here, we consider logarithmic size of memory (polynomial number of states) that is more common memory restriction for streaming models. In this case, quantum online algorithms with repeated test were considered in \cite{y2009}. In this paper, we focus on online streaming algorithms (one-way automata for online minimization problems) that read an input only once. This model was considered in \cite{kkkrym2017}. Authors show that quantum online streaming algorithms can be better than classical ones in the case of logarithmic size of memory. The model without ``one-way'' restriction but with sublogarithmic memory was considered in \cite{kk2019}. 
The question of comparing quantum and classical models was explored for streaming computation models (OBDDs and automata)\cite{l2009,gkkrw2007,agkmp2005,agky16,ss2005,kk2017,aakk2018}. 

 Here we focus on {\em advice complexity} measure \cite{k2016,bfklm2017,bkkkm2009,efkr2011,efkr2009,dkp2008}. In this case, an online algorithm gets some bits of advice about an input. A trusted {\em Adviser} sending these bits knows the whole input and has unlimited computational power. Deterministic and randomized online algorithms with advice are considered in \cite{h2005,k2016,bhkkr2012}. We compare the power of quantum online algorithms and classical ones in a case of using one-way automata with polynomial number of states (streaming algorithms with logarithmic size of memory). This question was not investigated before. 
Typically, the term ``Adviser'' is used in online algorithms theory; and the term ``Oracle'' for other models.  %

We consider the ``Black Hats Method'' for constructing hard online minimization problems from \cite{kkkrym2017,kk2019}. We use for construction  problems that separate power of quantum algorithms  from classical ones. Suppose that algorithms use only $O(\log n)$ bits of memory ($n^{O(1)}$ states), where $n$ is the length of the input.
%
There is a problem that has a quantum online streaming algorithm with a better competitive ratio than any classical (randomized or deterministic) online streaming algorithms, even if a classical one gets $o(\log n)$ advice bits. The problem is based on the $R$ function from \cite{ss2005}.
%
There is a problem that has quantum and randomized online algorithms with a better competitive ratio than any deterministic online algorithm, even if a deterministic one gets $o(\log n)$ advice bits. The problem is based on the Equality function and results from \cite{akv2008}.
%
For both cases, the quantum online streaming algorithms (with $O(\log n)$ qubits) have a better competitive ratio than any deterministic online algorithm with unlimited computational power.

%
Suppose that the algorithms use a constant size of memory (constant number of states). There is a problem that has the optimal online quantum streaming algorithm with $1$ qubit of memory and $1$ advice bit. A quantum online streaming algorithm with $1 $ qubit and without advice bits for the same problem has a better competitive ratio than any classical online streaming algorithm, even if it gets a non-constant number of advice bits. The problem is based on the $PartialMOD$ function from \cite{AY12,agky14} and the ``Black Hats Method''.
%
A modification of the problem has a quantum online streaming algorithm with a constant size of memory that has a better competitive ratio than any deterministic online algorithm with unlimited computational power, even if they get a constant number of advice bits and deterministic one has unlimited computational power.

  The paper is organized in the following way. Definitions are in Section \ref{sec:prlmrs}. The Black Hats Method is described in Section \ref{sec:bhm}. Quantum and randomized vs. deterministic online streaming algorithms are discussed in the first part of Section \ref{sec:application}; the second part contains results on quantum vs. classical models.
\section{Preliminaries}\label{sec:prlmrs}
%
{\bf An online minimization problem} consists of a set $\cal{I}$ of inputs and a cost function. Each input $I = (x_1, \dots , x_n)$ is a sequence of requests, where $n$ is a length of the input $|I|=n$. Furthermore, a set of feasible outputs (or solutions) ${\cal O}(I)$ is associated with each $I$; an output is a sequence of answers $O = (y_1, \dots, y_n)$. The cost function assigns a positive real value $cost(I, O)$ to $I\in{ \cal I}$ and $O\in{\cal O}(I)$. An optimal  solution for $I\in{\cal I}$ is $O_{opt}(I)=argmin_{O\in{\cal O}(I)}cost(I,O)$.

Let us define an online algorithm for this problem.
%
{\bf A deterministic online algorithm}  $A$ computes the output sequence $A(I) = (y_1,\dots , y_n)$ such that $y_i$ is computed by $x_1, \dots , x_i$.  
%
%
%
We say that $A$ is $c$-{\em competitive} if there exists a constant $\alpha\geq 0$ such that, for every $n$ and for any input $I$ of size $n$, we have: $cost(I,A(I)) \leq c \cdot cost(I,Opt(I)) + \alpha,$ where $Opt$ is an optimal offline algorithm for the problem and $c$ is the minimal number that satisfies the inequality. Also we call $c$ the {\bf competitive ratio} of $A$. If $\alpha = 0, c=1$, then $A$ is optimal.

   {\bf An online algorithm $A$ with advice} is a sequence of algorithms $A=(A^{0},\dots,A^{2^{b}-1})$ for some $b=b(n)$. The adviser chooses $\phi\in\{0,\dots,2^{b-1}\}$ that depends on an input $I$ and the algorithm $A^{\phi}$ computes an output sequence $A^{\phi}(I) = (y_1, \dots , y_n)$ such that $y_i=y_i(x_1, \dots , x_i)$.
  $A$ is $c$-competitive with advice complexity $b=b(n)$ if there exists a constant $\alpha\geq 0$ such that, for every $n$ and for any input $I$ of size $n$, there exists some $\phi\in\{0,\dots,2^{b-1}\}$ such that $cost(I,A^{\phi}(I)) \leq c \cdot cost(I,Opt(I)) + \alpha$.  
  
{\bf A randomized online algorithm} $R$ computes an output sequence
$R^{\psi}(I) = (y_1,\cdots, y_n)$ such that $ y_i$ is computed from $\psi, x_1, \cdots, x_i$, where $\psi$ is a content of the random tape, i. e., an infinite binary sequence, where every bit is chosen uniformly at random and independently of all the others. By $cost(I,R^{\psi}(I))$ we denote the random variable expressing the cost of the solution computed by $R$ on $I$.
$R$ is $c$-competitive in expectation if there exists a constant $\alpha\geq 0$ such that, for every $I$, $\exv[cost(I,R^{\psi}(I))] \leq c \cdot cost(I,Opt(I)) + \alpha$.

We use one-way automata for online minimization problems as online algorithms with restricted memory. In the paper, we use the terminology for branching programs \cite{Weg00} and algorithms. We say that an automaton computes Boolean function $f_m$ if for any input $X$ of length $m$, the automaton returns result $1$ iff $f(X)=1$. Additionally, we use the terminology on memory from algorithms. We say that an automaton has $s$ bits of memory if it has $2^s$ states. Let us present the definitions of automata that we use. A {\bf deterministic automaton} with $s=s(n)$ bits
of memory that process input $I=(x_1,\dots,x_n)$ is a $4$-tuple $(d_0,D,\Delta,Result)$. The set $D$ is a set of states, $|D|=2^s$, $d_0\in D$ is an initial state. $\Delta$ is a transition function $\Delta:\{0,\dots,\gamma-1\}\times D\to D$, where $\gamma$ is a size of the input alphabet. $Result$ is an output function $Result:D\to\{0,\dots,\beta-1\}$, where $\beta$ is a size of the output alphabet. The computation starts from the state $d_0$. Then on reading an input symbol $x_j$ it changes the current state $d\in D$ to $\Delta(x_j,d)$. In the end of computation, the automaton outputs $Result(d)$. A {\bf probabilistic automaton} is a probabilistic counterpart of the model. It
chooses from more than one transitions in each step such that each transition is associated with a probability. 
Thus, the automaton can be in a probability distribution over states during the computation. A total probability must be $1$, i.e., a probability of outgoing transitions from a single state must be $1$. Thus, a probabilistic automaton returns some result for each input with some probability. For $v\in\{0,\dots,\beta-1\}$, the automaton returns a result $v$ for an input, with bounded-error if the automata returns the result $v$ with probability at least $1/2 + \varepsilon$ for some $ \varepsilon \in (0,1/2] $. The automaton computes a function $f$ with bounded error  if it returns $f(X)$ with bounded error for each $X\in\{0,\dots,\gamma-1\}^n$. The automaton computes a function $f$ exactly if $\varepsilon=0$. 

A {\bf deterministic online streaming algorithm} with $s=s(n)$ bits
of memory that process input $I=(x_1,\dots,x_n)$ is a $4$-tuple $(d_0,D,\Delta,Result)$. The set $D$ is s set of states, $|D|=2^s$, $d_0\in D$ is an initial state. $\Delta$ is a transition function $\Delta:\{0,\dots,\gamma-1\}\times D\to D$. $Result$ is an output function $Result:D\to\{0,\dots,\beta-1\}$. The computation starts from the state $d_0$. Then on reading an input symbol $x_j$ it change the current state $d\in D$ to $\Delta(x_j,d)$ and outputs $Result(d)$.
A {\bf randomized online streaming algorithm} and a {\bf deterministic online streaming algorithm with advice} have similar definitions, but with respect to definitions of corresponding models of online algorithms.

{\bf Comment.} Note that any online algorithm can be simulated by online streaming algorithm using $n$ bits of memory.

Let us consider a {\bf quantum online streaming algorithm}. The good sources on quantum computation are \cite{nc2010,AY15}.
For some integers $n>0 $, a quantum online algorithm $ Q$ with $q$ qubits is defined
on input $I=(x_1,\dots,x_n)\in\{0,\dots,{\gamma-1}\}^n $ and outputs 
$(y_1,\dots,y_{n})\in\{0,\dots,\beta-1\}^{n}$.
A memory of the quantum algorithm is a state of a quantum register of $q$ qubits. In other words, the computation of $Q$ on an input $I$ can be traced by a $ 2^q$-dimensional vector from Hilbert space over the field of complex numbers. The initial state is a given $ 2^q$-vector $ |\psi\rangle_0$. In each step $j\in\{1,\dots,n\}$ the input variable $ x_{j} $ is tested and then a unitary $ 2^q\times2^q$-matrix $ G^{x_{j}}$ is applied:  
$
    |\psi\rangle_j = G^{x_{j}} (|\psi\rangle_{j-1}),
$ 
where  $ |\psi\rangle_j $ represents the state of the system after the $ j$-th step.
Depending on an input symbol, the algorithm can measure one or more quantum bits. If the outcome of the measurement is $u$, then the algorithm continues computing from a state $|\psi(u)\rangle$ and tha algorithm can output $Result(u)$ on this step. Here  $Result:
\{0,\ldots,2^{q}-1\}\to\{0,\dots, \beta-1\}$ is a function that converts the result of the measurement to an output variable. The algorithm $Q$ is $c$-competitive in expectation if there exists a non-negative constant $\alpha$ such that, for every $I$, $\mathbb{E}[cost(I,Q(I))] \leq c \cdot cost(I,Opt(I)) + \alpha$.

Let us describe a measurement process. Suppose that $Q$ is in a state $|\psi\rangle=(v_1, \dots, v_{2^q})$ before a measurement and the algorithm measures the $i$-th qubit. Suppose states with numbers $a^0_1,\dots,a^0_{2^{q-1}}$ correspond to $0$ value of the $i$-th qubit, and states with numbers $a^1_1,\dots,a^1_{2^{q-1}}$ correspond to $1$ value of the qubit. Then the result of the qubit's measurement is $1$ with probability $pr_1= \sum_{j=1}^{2^{q-1}}|v_{a^1_j}|^2$ and $0$ with probability $pr_0=1-pr_1$. If
the algorithm measures $v$ qubits on the $j$-th step, then $u \in\{0,\dots,2^v-1\}$ is an outcome of the measurement.

A quantum automata have the similar definition, but it returns $Result(u)$ in the end of the computation. A definition of a function computing is similar to the probabilistic case. See \cite{AY15} for more details on quantum automata.

In the paper we use results on id-OBDD. This model can be considered as an automaton with a transition function that depends on position of input head. You can read more about classical and quantum id-OBDDs in
\cite{Weg00,ss2005,agkmp2005,agky14,agky16,kk2017}. Formal definitions of
this model is in Appendix \ref{apx:obdd}. The following relations 
between models are folklore:

\begin{lemma}\label{lm:rel-obdd-sa}
 If a quantum (probabilistic) id-OBDD  $P$ of width $2^w$ computes a Boolean function $f$, then there is a quantum (probabilistic) automaton computing $f$ that uses $w + \lceil\log_2 n \rceil$ qubits (bits) of memory.
If any deterministic (probabilistic) id-OBDD  $P$ computing a Boolean function $f$ has a width at least $2^w$, then any deterministic (probabilistic) automaton computing $f$ uses at least $w$ bits of memory.
\end{lemma}

\section{The Black Hats Method for Constructing Online Minimization Problems}\label{sec:bhm}
Let us present a ``Black Hats Method''  which allows us to construct hard online minimization problems. It was defined in \cite{kkkrym2017}.
In this paper, we say a Boolean function $f$, but in fact we consider a family of Boolean functions $f=\{f_1,f_2,\dots\}$, for $f_m:\{0,1\}^m\to\{0,1\}$. We use notation $f(X)$ for $f_m(X)$ if the length of $X$ is $m$ and if it is clear from the context.

Suppose we have a Boolean function $f$ and positive integers $k,r,w,t$, where $k$ mod $t=0$, $r<w$. We define the online minimization problem $BH^t_{k,r,w}(f)$ as follows.
We have $k$ guardians and $k$ prisoners. They stay one by one in a line like $
G_1 P_1 G_2 P_2 \dots$, where $G_i$ is a guardian, $P_i$ is a prisoner. The prisoner
$P_i$ has an input $X_i$ of length $m_i$ and computes a function $f_{m_i}(X_i)$. If the
result is $1$, then the prisoner paints his hat black; otherwise, he paints it
white. Each guardian wants to know whether a number of following black hats is odd. We can say that the $i$-th guardian wants to compute $\bigoplus_{i=j}^k f_{m_i}(X_i)$. 
Formally, the problem is

{\bf Definition}[The Black Hats Method]
We have a Boolean function $f$. An online minimization problem $BH^t_{k,r,w}(f)$, for positive integers $k,r,w,t$, where $k$ mod $t=0$, $r<w$ is the following.
 Suppose we have an input $I=(x_1,\dots,x_n)$
and $k$ positive integers $m_1,\dots,m_k$, where $n=\sum_{i=1}^{k}(m_i+1)$. 
Let $I$ be such that $I=(2,X_1,2,X_2,2,X_3,2,\dots,2,X_k)$, where
$X_i\in\{0,1\}^{m_i}$, for $i\in\{1,\dots,k\}$. Let $O\in{\cal O}(I)$ and let
$O'=(y_1,\dots,y_k)$ be answer variables corresponding to input variables with
value $2$ (output variables for guardians). In other words, $y_j$ corresponds to $x_{i_j}$, where $i_j=j+\sum_{r=1}^{j-1}m_r$. Let 
$g_j(I)=\bigoplus_{i=j}^k f_{m_i}(X_i)$.
We separate all answer variables $y_i$ into $t$ blocks of length $ z=k/t $. A
cost of the $i$-th block is $c_i$. Here $c_i=r$ if $y_j=g_j(I)$, for $j\in\{(i-1)z+1,\dots,i\cdot z\}$; and $c_i=w$ otherwise.
The cost of the whole output is $cost^t(I,O)=c_1+\dots+c_t$.

If we have a quantum streaming algorithm for $f$ using a small amount of memory, then it is enough to guess the result of the first guardian to solve the problem.
Moreover, if there is no randomized streaming algorithm with small memory for $f$, then we cannot solve   $BH^t_{k,r,w}(f)$. The only way to reduce the competitive ratio is guessing the answers. 

Suppose we have a quantum streaming algorithm that uses small memory for $f$. Then long blocks can increase the gap between the competitive ratios of randomized and quantum algorithms because all guardians inside a block should return right answers.
We have a similar situation with deterministic online algorithms. In that case, we cannot guess answers; and we have a more significant gap between competitive ratios for quantum and deterministic algorithms. Advice bits can help for a classical algorithm. However, if we take sufficiently long blocks, then advice bits do not help.
The construction of the problem $BH^t_{k,r,w}(f)$ allows us to get a good competitive ratio by guessing only one bit; for example, this effect cannot be achieved by considering independent instances of the Boolean function $f$. 
These results are presented formally in theorems of this section. 

\begin{theorem}[\cite{kkkrym2017}]\label{th:bh-lower} Let $s$ be a positive integer, let $f$ be a Boolean function. Suppose there is no deterministic  automaton for $f$ that uses at most $s$ bits of memory. Then there is no $c$-competitive deterministic online streaming algorithm for $BH^t_{k,r,w}(f)$ that uses $s$ bits of memory, where $c<\frac{w}{r}$. 
\end{theorem}
\begin{theorem}[\cite{kkkrym2017}]\label{th:bh-rand-lower}
Let $s$ be a positive integer, let $f$ be a Boolean function. Suppose there is no probabilistic automaton that uses at most $s$ bits of memory and computes $f$ with bounded error. Then there is no $c$-competitive in expectation randomized online streaming algorithm $A$ for $BH^t_{k,r,w}(f)$ that uses $s$ bits of memory, where $c<2^{-z}+(1-2^{-z})w/r$, $z=k/t$.
\end{theorem}
\begin{theorem}[\cite{kkkrym2017}]\label{th:pqalgorithm} 
Let $s$ be a positive integer, let $f$ be a Boolean function. Suppose we have a quantum automaton $R$ that computes $f$ with bounded error $\varepsilon$ using $s$ qubits of memory, where $0\leq \varepsilon<0.5$. Then there is a quantum online streaming algorithm $A$ for $BH^t_{k,r,w}(f)$ that uses at most $s+1$ qubits of memory and has expected competitive ratio %
$c\leq \left(0.5(1-\varepsilon)^{z-1}\cdot(r-w) + w\right)/r$, $z=k/t$. 
\end{theorem}

For proofs of the following properties for classical models, we show that if the model does not have enough memory, then the problem can be interpreted as the ``String Guessing, Unknown History'' problem from \cite{bhkkss2014}. 
\begin{theorem}\label{th:advice}
Let $s$ be a positive integer, let $f$ be a Boolean function. Suppose there is no deterministic automaton for $f$ that uses at most $s$ bits of memory. Then there is no $c$-competitive deterministic online streaming algorithm for $BH^t_{k,r,w}(f)$ that uses $s$ bits of memory and $b$ advice bits, where $c<(hr + (t-h)w)/(tr)$, $h=\lfloor v/z \rfloor, z=k/t$, $v$ is such that $b=(1 + (1-v/k) \log_2(1-v/k)+$ $(v/k) \log_2 (v/k))k$, $0.5\cdot k\leq v <k$.
\end{theorem}
\Beginproof
Let us prove the following claim. If the online algorithm gets $b$ advice bits, then there is an input such that at least $k-b$ prisoners return wrong answers.
We prove it by induction.

Firstly, let us prove the claim for $b=k$. Then the adviser can send $(g_1,\dots,g_k)$, for $g_i=\bigoplus_{j=i}^k f(X^j)$. Then the algorithm can return right answers for all guardians.

Secondly, let us prove the claim for $b=0$. It means that the algorithm does not have any advice and we get the situation from Theorem \ref{th:bh-lower}.

Thirdly, let us consider the claim for other cases. Assume that we already proved the claim for any pair $(b'',k')$ such that $b'\leq b$, $k' \leq k$ and at least one of these inequalities is strict. We focus on the first prisoner. 

Assume that there is an input $X^1\in\{0,1\}^{m_1}$ for the first prisoner such that this prisoner cannot compute an answer with bounded error. Then we use this input and get a situation for $k-1$ prisoners and $b$ advice bits. In that case, $k-b-1$ prisoners are wrong, plus the first one is also wrong.

Assume that the algorithm always can compute an answer with a bounded error for the first prisoner.
So we can describe the process of communication with the adviser in the following way:
the adviser separates all possible inputs into $2^b$ non-overlapping groups $G_1,\dots,G_{2^b}$. After that, he sends a number of the group that contains current input to the algorithm. Then the algorithm $A$ processes the input with the knowledge that an input can be only from this group.

Let us consider three sets of groups:
\begin{itemize}
\item 
$I_0=\{G_i: \forall \sigma\in\{0,1\}^{m_1}$ such that $\sigma$ is an input for the first prisoner and $f(\sigma)=0\}$, 
\item 
$I_1=\{G_i: \forall \sigma\in\{0,1\}^{m_1}$ such that $\sigma$ is an input for the first prisoner and $f(\sigma)=1\}$, 
\item 
$I_{10}=\{G_1,\dots,G_{2^b}\}\backslash(I_1 \cup I_0)$.
\end{itemize} 

Let $|I_a|\neq 0$, for some $a\in\{0,1\}$. If $|I_a|\leq 2^{b-1}$, then as $X^1$ we take any input from any group $G\in I_a$. Hence we have at most $ 2^{b-1}$ possible groups for the adviser that distinguish inputs of next guardians. We can say that the adviser can encode them using $b-1$ bits. Therefore, we get the situation for $b-1$ advice bits and $k-1$ prisoners. The claim is true for this situation.
If $|I_a|> 2^{b-1}$, then we take any input from any group $G\not\in I_{a}$ as $X^1$. Hence, we have at most $ 2^{b-1}$ possible groups for the adviser and the same situation. The claim is true for this situation.

Let $|I_0|=|I_1|=0$. Suppose that the randomized online algorithm can solve the problem using $s'$ bits of memory, where $s'<s-b$. We can simulate the work of the algorithm with advice on $X^1$ using the automaton $B$ with the following structure. $B$ has two parts of memory: $M_1$ of $b$ bits and $M_2$ of $s'$ bits. Suppose that the adviser initialized $M_1$ by advice bits. Then $B$ invokes $A$ depends on the value of $M_1$ and advice bits. So, $B$ can simulate the work of $A$, the automaton $B$ uses $s'+b<s$ bits of memory and computes $f$.  It is a contradiction with the claim of the theorem. Therefore, the only way to compute the result for the first prisoner is sending answer as one advice bit. So, we have a situation for $k-1$ prisoners and $b-1$ advice bits.

So, it means, that for the algorithm the problem is the same as the {\em String Guessing Problem with Unknown History}($2\!-\!SQUH$) from \cite{bhkkss2014}.

The following result for the $2\!-\!SQUH$ is known:
\begin{lemma}[\cite{bhkkss2014}]
Consider an input string of length $k$ for $2\!-\!SGUH$, for some positive integer $k$. Any online algorithm that is correct in more than $\alpha k$
characters, for $0.5 \leq \alpha < 1$, needs to read at least
$\left(1 + (1-\alpha) \log_2(1-\alpha)+ \alpha \log_2 \alpha\right)k$
advice bits.
\end{lemma}

Therefore, if we want to get $v$ right answers for guardians, then we need 

$b=\left(1 + (1-\frac{v}{k}) \log_2(1-\frac{v}{k})+ \frac{v}{k} \log_2 \frac{v}{k}\right)k$.

 Because of properties of the cost function, the best case for the algorithm is getting right results about all guardians of a block. Hence, the algorithm can get $h=\lfloor v/z \rfloor$ full blocks and the cost for each of them will be $r$, for $z=k/t$. Other blocks have at least one ``wrong'' guardian, and these blocks cost $w$.
Therefore, we can construct an input such that it costs $\lfloor v/z \rfloor \cdot r + (t-\lfloor v/z \rfloor)w$, for $b=\left(1 + (1-\frac{v}{k}) \log_2(1-\frac{v}{k})+ \frac{v}{k} \log_2 \frac{v}{k}\right)k$. Hence the competitive ratio $c$ of the algorithm is $c\geq \frac{\lfloor v/z \rfloor \cdot r + (t-\lfloor v/z \rfloor)w}{tr}$, for $b=\left(1 + (1-\frac{v}{k}) \log_2(1-\frac{v}{k})+ \frac{v}{k} \log_2 \frac{v}{k}\right)k$.
\Endproof
%
%

\begin{corollary}\label{cr:nobenefit-advice}
Let $s$ be a positive integer. Suppose a Boolean function $f$ is such that no deterministic streaming algorithm uses at most $s$ bits of memory and computes $f$. Then there is no $c$-competitive deterministic online streaming algorithm that uses $s$ bits of memory and $b$ advice bits, and solves $BH^t_{k,r,w}(f)$, where $c<w/r$, $v<z$.
\end{corollary}

Let us present a randomized analog of Theorem
\ref{th:advice}. The proof is based on ideas
from \cite{kiy2018,DS90,AK13,Sh59}. We use a function $\delta_x$ in the claim of
the following theorem: $\delta_x=1$ if $x\neq 0$; otherwise, $\delta_x=0$.

\begin{theorem}\label{th:radvice}
Let $s$ be a positive integer, let $f$ be a Boolean function. Suppose there are no probabilistic automaton that compute $f$ with bounded error using space less than $s$ bits. Then any
randomized online streaming algorithm $A$ using less than $s-b$ bits, $b$
advice bits and solving $BH^t_{k,r,w}(f)$, has the expected competitive ratio 
$c\geq ( hr + \delta_{u}\cdot (2^{u-z}r + (1-2^{u-z})w) + (t-h-\delta_{u})(2^{-z}r+(1-2^{-z})w))/(tr)$, for $h=\lfloor v/z \rfloor, z=k/t, u =
v-hz$, $v$ is such that $b=\left(1 + (1-v/k) \log_2(1-v/k)+ (v/k) \log_2 (v/k)\right)k$, $0.5k\leq v <k$. 
\end{theorem}
\Beginproof
The proof is similar to the proof of Theorem \ref{th:advice}. If the online algorithm gets $b$ advice bits, then there is an input such that at least $k-b$ prisoners return wrong answers. We can prove this claim by the same way as in Theorem \ref{th:advice}, but we use the probabilistic automaton $B$ for simulating $A$ because all parts of automaton that use memory $M_2$ are probabilistic.
So, it means, that for the algorithm the problem is the same as the {\em String Guessing Problem with Unknown History}($2\!-\!SQUH$) from \cite{bhkkss2014}.
Therefore, if we want to get $v$ right answers for guardians, then we need $b=\left(1 + (1-\frac{v}{k}) \log_2(1-\frac{v}{k})+ \frac{v}{k} \log_2 \frac{v}{k}\right)k$.
We can show that the guardian that does not get an answer from the adviser (``unknown'' guardians)  cannot be computed with bounded error. Therefore, they can be only guessed with probability $0.5$. We can use the proof technique as in Theorem \ref{th:bh-rand-lower}. We use the same approach for all segments between ``known'' guardians.
\Endproof

If we have a quantum streaming algorithm for $f$; then one advice bit is enough for solving $BH^t_{k,r,w}(f)$ with a good competitive ratio:
\begin{theorem}\label{th:pqalgorithm-advice}
Let $s$ be a positive integer, let $f$ be a Boolean function. Suppose there is a quantum automaton $R$ that computes $f$ with bounded error $\varepsilon$ using $s$ qubits of memory, where $0\leq \varepsilon<0.5$. Then there is a quantum  online streaming algorithm $A$ using at most $s+1$ qubits of memory, single advice bit and solving $BH^t_{k,r,w}(f)$ such that the expected competitive ratio is $c\leq \left(0.5(1-\varepsilon)^{z-1}\cdot \left(t +1 + \frac{v^t-v}{v-1} \right)(r-w) + tw\right)/(tr)$, for $v=(1-2\varepsilon)^z, z=k/t$.
If $\varepsilon=0$, then $c=1$. 
\end{theorem}
\Beginproof
Let us describe the quantum online streaming algorithm $A$:

{\bf Step 1.} Algorithm $A$ gets $g_1=\bigoplus\limits_{i=1}^k f(X^i)$ as an advice bit. The algorithm stores current result in qubit $|p\rangle$. Then the algorithm measures $|p\rangle$ and gets $g_1$ with probability $1$. Then $A$ returns the result of the measurement as $y_1$.

{\bf Step 2.} The algorithm reads $X^1$ and computes $|p\rangle$ as a result of  CNOT or XOR of $|p\rangle$ and  $R(X^1)$, where $R(X^1)$ is the result of computation for $R$ on the input $X^1$. $A$ uses register $|\psi\rangle$ of $s$ qubits for processing $X^1$. Then the algorithm returns a result of a measurement for $|p\rangle$ as $y_2$. After that $A$ measures all qubits of $|\psi\rangle$ and sets $|\psi\rangle$ to $|0\dots 0\rangle$. The algorithm can do it because it knows a result of the measurement and can rotate each qubit such that the qubit becomes $|0\rangle$. 

{\bf Step $i$.} The algorithm reads $X^{i-1}$ and computes $|p\rangle$ as a result of  CNOT or XOR of $|p\rangle$ and  $R(X^{i-1})$.  Algorithm $A$ uses register $|\psi\rangle$ of $s$ bits on processing $X^{i-1}$. Then $A$ returns a result of the measurement for $|p\rangle$ as $y_i$. After that $A$ measures $|\psi\rangle$ and sets $|\psi\rangle$ to $|0\dots 0\rangle$

Let us compute a cost of the output for this algorithm.
Let us consider a new cost function $cost'(I,O)$. For this function, a ``right''  block costs $1$ and a ``wrong'' block costs $0$. In that case, $cost^t(I,O)=(r-w)\cdot cost'(I,O) + tw$. Therefore, in the following proof we can consider only the $cost'(I,O)$ function.

Firstly, let us compute $p_i$ that is the probability that block $i$ is a ``right'' block (or costs $1$). 
Let $i=1$. So, if the $i$-th block is ``right'', then all $z-1$  prisoners inside the block return right answers. A probability of this event is $p_1=(1-\varepsilon)^{z-1}$.

Let $i>1$. If the $i$-th block is ``right'', then two conditions should be true:

(i) All $z-1$ prisoners inside the block should return right answers. 

(ii) If we consider a number of preceding guardians that return wrong answers plus $1$ if the preceding prisoner has an error. Then this number should be even.

A probability of the first condition is $(1-\varepsilon)^{z-1}$. Let us compute a probability of the second condition.

Let $E(j)$ be the number of errors before the $j$-th guardian. It is a number of errors for previous prisoners. Let $F(j)$ be a probability that $E(j)$ is even. Therefore, $1-F(j)$ is a probability that $E(j)$ is odd.

If there is an error in a computation of the $(j-1)$-th prisoner, then $E(j-1)$ should be odd. If there is no error for the $(j-1)$-th prisoner, then $E(j-1)$ should be even. Hence, $F(j)=\varepsilon (1-F(j-1)) + (1-\varepsilon)F(j-1)= F(j-1)\cdot(1-2\varepsilon) + \varepsilon$.  Note that $F(1)=1$, because the first guardian gets the answer as the advice bit. 

$F(j)=F(j-1)\cdot(1-2\varepsilon) + \varepsilon = F(j-2)\cdot(1-2\varepsilon)^2 + (1-2\varepsilon)\varepsilon + \varepsilon=\dots =$

$=F(j-j+1)\cdot(1-2\varepsilon)^{j-1} + (1-2\varepsilon)^{j-2}\varepsilon + \dots (1-2\varepsilon)\varepsilon + \varepsilon=F(1)\cdot(1-2\varepsilon)^{j-1} + \varepsilon\sum_{l = 0}^{j-2}(1-2\varepsilon)^{l}=$
$(1-2\varepsilon)^{j-1} +\frac{1-(1-2\varepsilon)^{j-1}}{2}=\frac{1+(1-2\varepsilon)^{j-1}}{2}$

So,
$F(j)=0.5\cdot((1-2\varepsilon)^{j-1} +1)$.

The probability $p_i$ of the event is:
$p_i =(1-\varepsilon)^{z-1}\cdot\frac{1+(1-2\varepsilon)^{(i-1)z+1-1}}{2}$

So the expected cost is
$\exv[cost'(I,O)]=\sum_{i=1}^{t}\big(p_i\cdot 1 + (1-p_i)\cdot 0\big)=p_1+\sum_{i=2}^tp_i$

$\exv[cost'(I,O)]=(1-\varepsilon)^{z-1}\cdot\left(1 +\sum_{i=2}^{t}(0.5 + 0.5(1-2\varepsilon)^{(i-1)z})\right)=$

$= 0.5(1-\varepsilon)^{z-1}\cdot \left(t +1 +\sum_{i=2}^{t}(1-2\varepsilon)^{(i-1)z} \right)=$
$=0.5(1-\varepsilon)^{z-1}\cdot \left(t +1 + \sum_{i=2}^{t}\left((1-2\varepsilon)^z\right)^{i-1} \right)=$

$=0.5(1-\varepsilon)^{z-1}\cdot \left(t + 1 + (v^t-v)/(v-1) \right)$,
where $v=(1-2\varepsilon)^z$. If $\varepsilon=0$, then $\exv[cost'(I,O)]=t$.

Therefore,
$\exv[cost^t(I,O)]=0.5(1-\varepsilon)^{z-1}\cdot \left(t +1 + (v^t-v)/(v-1) \right)(r-w) + tw$, for $v=(1-2\varepsilon)^z$.
If $\varepsilon=0$, then $\exv[cost^t(I,O)]=tr$.
\Endproof
%

We can also modify the method to capture the case of several advice bits.

{\bf Definition} [The Interleaved Black Hats method]
  Let $f$ be a non-constant Boolean function. 
  Online minimization problem $IBH^{\lambda,t}_{k,r,w}(f)$ for integers $k, r, w, t, \lambda$ where $k \mod t = 0,\lambda>0$ is the following.
  %
%
  Suppose we have input $I \in \{0,1,2\}^n$ such that 
 $
  I = (2,X^1_1, 2,\ldots, 2,X^\lambda_1,
  2, \ldots,
  2,X^1_k,$ $\ldots, 2,X^\lambda_k)
  ,$
  where $X^j_i \in \{0,1\}^{m_i}$ for $i \in \{1,
  \ldots, k\}$, \( j \in \{1, \ldots, \lambda\} \), and $n = \lambda \sum_{i=1}^{k} (m_i + 1)$.
  Let $O\in{\cal O}(I)$ and $O' = (y^1_1 \ldots y_k^\lambda)$ be answer bits
  corresponding to input variables with value $2$ (guardians).
  Let $g^j_i(I) = \bigoplus_{a=i}^k f(X^j_a)$, where $i \in \{1, \ldots, k\}$ and $j \in \{1, \ldots, \lambda\}$.
  We split all output variables $y^j_i$ into $t$ blocks of length $z =
  \lambda k/t$. The cost of the $l$-th block is $c_l^j$, where $c_l^j = r$ if
  $y^j_i = g^j_i(I)$ for all \( i \in \{ (l-1)z+1, \ldots, l\cdot z \} \) and \( j \in \{1,\ldots, \lambda\} \),
  and $c_l^j = w$ otherwise. The cost of the whole output is $cost'(I,O) = \sum_{j=1}^{\lambda}\sum_{l=1}^{k/z} c^j_l$. Note that $IBH^{1,t}_{k,r,w}(f) = BH^t_{k,r,w}(f)$.


\begin{theorem}\label{th:ibh-lam}

Let $s$ be a positive integer, let $f$ be a Boolean function.
 Suppose there is a quantum automaton $Q$ that computes $f$ exactly (with zero-error) using $s$ qubits of memory. Then there is a quantum online streaming algorithm $A$ that solves \(IBH^{\lambda,t}_{k,r,w}(f)\) using at most \(s + \lambda\) qubits of memory and \(\lambda-1\) advice bits. The algorithm has expected competitive ratio  \(c=\frac{r+w}{2r}\).
There is no deterministic online algorithm $B$ computing
$IBH^{\lambda,t}_{k,r,w}(f)$ that uses $\lambda-1$ advice bits and has competitive ratio $c < w/(tr) + (t-1)/t$. 
\end{theorem}
\Beginproof
The algorithm $A$ is similar to the algorithm from the proof of Theorem
\ref{th:pqalgorithm}.  Let us give a brief description of the algorithm \(A\).
The only step of the algorithm that is different from others is step 1.1: on
this step algorithm guesses a value of \(y^1_1\) and stores it in qubit $|p_1\rangle$. Steps 1.j use
an advice bit learning a value \( y^j_1 \) for \( j \in \{2,\ldots,\lambda\} \) from
the adviser and storing the value in a qubit $|p_1\rangle$. Steps 2.1--\(k\).\(\lambda\) use the corresponding \( X^j_i \) to
update the value of \(|p_i\rangle\) and to return \( y^j_i \).

{\bf Step 1.1.} The algorithm \(A\) guesses \(y^1_1\) with
equal probabilities and stores it in a qubit \(|p_1\rangle\): the algorithm initialize the qubit \(|p_1\rangle=\frac{1}{\sqrt{2}}|0\rangle+\frac{1}{\sqrt{2}}|1\rangle\). Then \(A\) measures $|p_1\rangle$ and returns a result of the measurement as \( y^1_1\).

{\bf Step 1.2.} The algorithm $A$ gets $g_1^2=\bigoplus_{i=1}^k f(X^2_i)$ as an advice bit
and stores it in a qubit $|p_2\rangle$. Then $A$ measures the qubit and  returns a result of the measurement as \( y^2_1\).

{\bf Step 1.3.} The algorithm $A$ gets $g^3_1=\bigoplus_{i=1}^k f(X^3_i)$ an advice bit and stores it in a qubit $|p_3\rangle$. Then $A$ measures the qubit and  returns a result of the measurement as \( y^3_1\).

\ldots

{\bf Step 1.$\lambda$.} The algorithm $A$ gets $g^{\lambda}_1=\bigoplus_{i=1}^k f(X^{\lambda}_i)$ as an advice bit
and stores it in a qubit $|p_\lambda\rangle$. Then $A$ measures the qubit and returns a result of the measurement as \( y^\lambda_1\).

{\bf Step 2.1.} The algorithm $A$ reads $X^1_1$ and computes $|p_1\rangle$ as a result of CNOT or XOR gate for $|p_1\rangle$ and $R(X^1_1)$. Then $A$ measures the qubit and returns a result of the measurement as \( y^1_{2}\). Here $R(X^1_1)$ is a result of computing $R$ on $X^1_1$. $A$ uses register $|\psi\rangle$ of $s$ qubits for processing $X^1_1$. In the end of the step $A$ measures all qubits of $|\psi\rangle$ and sets $|\psi\rangle$ to $|0\dots 0\rangle$. The algorithm can do it, because it knows a result of the measurement and can rotate each qubit such that the qubit becomes $|0\rangle$. 

\ldots

{\bf Step $i$.$j$.} 
The algorithm reads $X^{j}_{i-1}$ and computes $|p_j\rangle$ as a result of CNOT or XOR gate for $|p_j\rangle$ and $R(X^{j}_{i-1})$.
Then $A$ measures the qubit and returns a result of the measurement as \( y^j_i\). In the end of the step $A$ measures all qubits of $|\psi\rangle$ and sets $|\psi\rangle$ to $|0\dots 0\rangle$. 

\ldots

{\bf Step $k$.$\lambda$.} 
The algorithm reads $X^{\lambda}_{k-1}$ and computes $|p_\lambda\rangle$ as a result of CNOT or XOR gate for $|p_j\rangle$ and $R(X^{\lambda}_{k-1})$.
Then $A$ measures the qubit and returns a result of the measurement as \( y^{\lambda}_{k}\).

Let us consider any input $I$.
If the first guardian of the first instance $G^1_1$ guesses the right answer, then all guardians of the first instance $G^1_i$ are also right, for $i\in \{1,\dots, k\}$. If $G^1_1$ guesses the wrong answer, then all $G^1_i$ are wrong.
Guardians of other instances are always right, because the algorithm gets $g^j_1$ as advice bits, for \(j\in \{2,\dots, \lambda\}\).

Let us compute an expected cost of an output $A(I)$ for the input $I$.
We know that $k \mod t=0$. Therefore, the following statement holds: $u\leq \lambda$, where $u=\lambda k/t$. Hence, each block contains at least one guardian of the first instance. If the guess of the guardian $G^1_1$ is wrong, then each block contains at least one wrong guardian and the cost of the input is $wt$. If the guess of the guardian $G^1_1$ is right, then all guardians return right answers and the cost of the input is $rt$. The guardian $G^1_1$ chooses $0$ or $1$ with equal probability. Therefore, \(\exv{\mathrm{cost}(I,A(I))}=0.5rt + 0.5wt\).
So, expected competitive ratio is $c=\frac{0.5rt + 0.5wt}{rt}=\frac{r + w}{2r}$.

Let us show nonexistence Of Algorithm $B$. We want to prove that any deterministic algorithm fails even if we give
$\lambda-1$ advice bits.
Suppose that there exists some algorithm B that uses $\lambda-1$ advice bits
and is $c$-competitive for some $c < w/(tr) + (t-1)/t$.
Let us consider inputs ${\cal I}$ with the following property. Any $I\in {\cal I}$ is such that $I=2,X^1_1\dots,2,X^{\lambda}_1,2\dots,2,X^1_k\dots,2,X^{\lambda}_k$, where $f(X^j_i)=0$, for any $j\in\{1,\dots,\lambda\}$, $i\in\{1,\dots,k-1\}$. 
Let $\sigma=(\sigma_1,\dots,\sigma_{\lambda})\in\{0,1\}^{\lambda}$. Let input $I_\sigma\in{\cal I}$ be such that $f(X^j_k)=\sigma_j$, for $j\in\{1,\dots,\lambda\}$.

The adviser sends $\lambda-1$ bits of advice to the algorithm $B$, and this advice separates all possible inputs into $2^{\lambda-1}$ groups.
By Pigeonhole Principle there are $\sigma,\sigma'\in\{0,1\}^{\lambda}$, such that $\sigma\neq \sigma'$ and they both belong to the same group of the adviser. Let integer $\kappa$ be such that $\sigma_j=\sigma'_j$ for $j<\kappa$ and $\sigma_\kappa \neq \sigma'_\kappa$

The results $y^\kappa_k$ of algorithm $B$ does not depends on $X^\kappa_k$, because it is returned before reading this part of the input. It means that $y^\kappa_k$ depends only on the group and previous input. 
Therefore the algorithm $B$ returns the same $y^\kappa_k$ on $I_{\sigma}$ and $I_{\sigma'}$. At the same time, $g^\kappa_k(I_\sigma)=\sigma_\kappa\neq \sigma'_\kappa=g^\kappa_k(I_{\sigma'})$. Therefore, the algorithm returns wrong answer at least on one of guardians for $I_{\sigma}$  or $I_{\sigma'}$. Let this input be $I_w$.

As the IBH problem requires correct answers for each of guardians, the whole block costs $w$.
Finally, the cost of algorithm $B$ for input $I_w$ is at least $w + (t-1)r$ and it is
\(c_0\)-competitive, for some \( c_0 \geq w/(tr) + (t-1)/t \). 
We get a contradiction.
\Endproof

\section{Application}\label{sec:application}
Let us discuss the applications of the Black Hats Method. In this section, we present examples of problems that allow us to show the benefits of quantum computing in the case of online streaming algorithms. 
Here we use results for OBDDs (See Appendix \ref{apx:obdd})
Recall that $BH^t_{k,r,w}(f)$ is a black hat problem for $k$ guardians, $t$ blocks of guardians, $r$ and $w$ are costs for a right and a wrong answers of a block, respectively, $z=k/t$ and $k$ mod $t=0$.

\subsection{Polylogarithmic Size of Memory}
We start by analyzing the model with polylogarithmic size of memory.
Let us apply the Black Hats Method from Section \ref{sec:bhm} to  a Boolean function $R_{\nu, l, m,u}:\{0,1\}^n\to\{0,1\}$ from \cite{ss2005}:
Let $|1\rangle, \dots , |u\rangle$ be the standard basis of $\mathbb{C}^u$. Let
$V_0$ and $V_1$ denote the subspaces spanned by the first and last $u/2$ of these basis vectors.
Let $0 < \nu < 1/\sqrt{2}$. The input for the function $R_{\nu, l, m,u}$  consists of $3l(m + 1)$ Boolean variables $a_{i,j}, b_{i,j}, c_{i,j}, 1 \leq i \leq l , 1 \leq j \leq m+1$, which are interpreted as universal $(\epsilon, l, m )$-
codes for three unitary $u \times u$-matrices A, B, C, where $\epsilon = 1/(3u)$. The function takes the
value $z \in \{0, 1\}$ if the Euclidean distance between $CBA|1\rangle$ and $V_z$ is at most $\nu$. Otherwise the function is undefined.
It is known from \cite{ss2005} that there is a quantum OBDD that computes $R_{\nu, l, m,u}$ using linear width.  At the same time, any deterministic or probabilistic OBDD requires exponential width. Therefore, we have the following result due to Lemma \ref{lm:rel-obdd-sa}.

\begin{lemma}\label{lm:bhr}
1. There is a quantum automaton that computes $R_{\nu, l, m,u}$ with bounded error $\nu^2$. using $O(\log n)$ qubits. 
2. There is no probabilistic automaton that computes $R_{\nu, l, m,u}$ with bounded error using $n^{o(1)}$ bits of memory .
\end{lemma}
Let us consider the $BHR^t_{k,r,w,\nu, l, m,u} =BH^t_{k,r,w}(R_{\nu, l, m,u})$ problem.
The following properties of the problem are based on Lemma \ref{lm:bhr} and Theorems  \ref{th:bh-lower}, \ref{th:bh-rand-lower}, \ref{th:pqalgorithm}- \ref{th:pqalgorithm-advice}. 

\begin{theorem}
Suppose $P^t=BHR^t_{k,r,w,\nu, l, m,u}$,  $t\in\{1,\dots,k\}$, $k = (\log_2 n)^{O(1)}$,  $v$ is such that

$b=\left(1 + (1-v/k) \log_2(1-v/k)+ (v/k) \log_2 (v/k)\right)k$, $0.5k\leq v <k$; then

1. There is no $c$-competitive deterministic online streaming algorithm with $n^{o(1)}$ bits of memory and $b$ advice bits that solves $P^t$, where $c<{\cal C}_1=\frac{w}{r}$, $v<z$. There is no $c$-competitive deterministic online streaming algorithm that uses $n^{o(1)}$ bits of memory and $b$ advice bits, and solves $P^t$, where $c<{\cal C}_2=\frac{hr + (t-h)w}{tr}$, $h=\lfloor v/z \rfloor$.

2. There is no randomized streaming algorithm using $n^{o(1)}$ bits of memory, $b<k$ advice bits and solving $P^t$ that is $c$-competitive for 
$c<{\cal C}_4=\frac{ hr + \delta_{u}\cdot (2^{u-z}r + (1-2^{u-z})w) + (t-h-\delta_{u})(2^{-z}r+(1-2^{-z})w)}{tr}$, $h=\lfloor v/z \rfloor, u =
v-hz$.

3. There is a quantum online streaming algorithm that uses $O(\log n)$ qubits and solves $P^t$. The algorithm $Q$ is $c$-competitive in expectation, where
$c\leq \left(\left(1-\nu^2\right)^{z-1}\cdot 0.5 \cdot (r-w) + w\right)/r<{\cal C}_1,{\cal C}_2,{\cal C}_3,{\cal C}_4$.


\end{theorem}

This theorem gives us the following important results.
There is a quantum online streaming algorithm with logarithmic size of memory for  $BHR^t_{k,r,w,\nu, l, m,u}$ having a better competitive ratio than
 any classical (deterministic or randomized) online streaming algorithm with polylogarithmic size of memory, even if they use a polylogarithmic number of advice bits

\subsection{Sublogarithmic Size of Memory}
We continue by analyzing the model with sublogarithmic memory.
Let us discuss the $PartialMOD_m^{\beta}$ function from \cite{AY12,agky14,agky16}. Feasible inputs for the problem are $X\in\{0,1\}^n$ such that $\#_1(X)=v\cdot 2^{\beta}$, where $\#_1(X)$ is the number of $1$s and $v\geq 2$. $PartialMOD_m^{\beta}(X)=v$ $\mod$ $2$.
It is known from \cite{AY12,agky14,agky16} that there is a quantum automaton that computes $PartialMOD_m^{\beta}$ using a single qubit and has not error. At the same time, any deterministic or probabilistic automaton and id-OBDDs requires $2^{\beta}$ states (width). Hence, we have the following result due to Lemma \ref{lm:rel-obdd-sa}.

\begin{lemma}\label{lm:bhp}
1. There is a quantum automaton that computes $PartialMOD^{\beta}_m$ exactly  using $1$ qubit;
2. There is no probabilistic automaton that computes $PartialMOD^{\beta}_m$  with bounded error using less than $\beta$ bits.
\end{lemma}

Let us apply the Black Hats Method to $f=PartialMOD_m^{\beta}$.  The proof of the following theorem is based on Theorems \ref{th:pqalgorithm-advice},\ref{th:advice},\ref{th:bh-lower},\ref{th:ibh-lam} and Lemma \ref{lm:bhp}.

\begin{theorem}
Suppose $P^t=BHM^t_{k,r,w}=BH^t_{k,r,w}(PartialMOD^{\beta}_m)$, $t\in\{1,\dots,k\}$, $\beta=O(\log n)$, $k=o(\log n), \beta \cdot k<\log_2 n$,  $v$ is such that
%
$b=\left(1 + (1-v/k) \log_2(1-v/k)+ (v/k) \log_2 (v/k)\right)k$, $0.5k\leq v <k$; then

1. There is no deterministic online streaming algorithm using $s<\beta$ bits of memory, $b <k$ advice bits and solving $P^t$ that is $c$-competitive for $c<{\cal C}_2=\frac{hr + (t-h)w}{tr}$, $h=\lfloor v/z \rfloor$.

2. There is no randomized streaming algorithm using $s<\beta$ bits of memory, $b <k$ advice bits and solving $P^t$ that is $c$-competitive for 
$c<{\cal C}_5=\frac{ hr + \delta_{u}\cdot (2^{u-z}r + (1-2^{u-z})w) + (t-h-\delta_{u})(2^{-z}r+(1-2^{-z})w)}{tr}$, $h=\lfloor v/z \rfloor, u =
v-hz$.


3. There are quantum online streaming algorithms $A$ and $B$ for $P^t$. The algorithm $A$ gets $1$ advice bit, uses $1$ qubit of memory and $A$ is optimal. The algorithm $B$ does not get advice bits, uses $1$ qubit of memory and has expected competitive ratio $c\leq\frac{r+w}{2r}<{\cal C}_1,{\cal C}_2,{\cal C}_5$.

4. Suppose $IBHM=IBH^{\lambda,1}_{k,r,w}(PartialMOD^{\beta}_m)$, $\lambda=const$. Then there is a quantum online streaming algorithm $Q'$ with $\lambda$ qubits and $\lambda-1$ advice bits such that the algorithm computes $IBHM$ with expected competitive ratio $c\leq \frac{w+r}{2r}<{\cal C}_1$. Any deterministic online algorithm for $IBHM$ with unlimited computation power has competitive ratio $c = {\cal C}_1=w/r$. 
\end{theorem}
\Beginproof
 Claims 1, 2 and 4 follow from Theorems \ref{th:pqalgorithm-advice},\ref{th:advice},\ref{th:bh-lower},\ref{th:ibh-lam} and Lemma \ref{lm:bhp}.
Let us prove Claim 3.
The proof of the theorem follows from Theorems 1, 11 and Algorithm 1 from \cite{kkm2018}.
Let us describe an algorithm $B$ for $BH^t_{k,r,w}(f)$, $f=PartialMOD_m^\beta$.

{\bf Step $1$.} The algorithm emulates guessing for $g_1=\bigoplus\limits_{j=1}^{k}f(X^j)$. $B$ starts on a state $|\psi\rangle=\frac{1}{\sqrt{2}}|0\rangle+\frac{1}{\sqrt{2}}|1\rangle$. The algorithm measures the qubit $|\psi\rangle$ before reading any input variables. $B$ gets $|0\rangle$ or $|1\rangle$ with equal probabilities. The result of the measurement is $y_1$.

{\bf Step $2$.} The algorithm reads $X^1$.  Let an angle $\xi=\pi/2^{\beta+1}$. Algorithm $B$ rotates the qubit by the angle $\xi$ if the algorithm meets $1$. It does not do anything for $0$.

{\bf Step $3$.} If $B$ meets $2$, then it measures the qubit $|\psi\rangle=a|0\rangle+b|1\rangle$. If $PartialMOD_{m_1}^\beta(X^1)=1$, then the qubit is rotated by an angle $\pi/2+v\cdot \pi$, for some integer $v$, else the qubit is rotated by an angle $w\cdot \pi$, for some integer $w$. If $y_1=1$, then $a\in\{1,-1\}$ and $b=0$. If $y_1=0$, then $a=0$ and $b\in\{1,-1\}$. The result of the measurement for the qubit $|\psi\rangle$ is $y_2$. 

{\bf Step $4$.} The algorithm reads $X^2$. Algorithm $B$ rotates the qubit $|\psi\rangle$ by the angle $\xi$ if the algorithm meets $1$. It does not do anything for $0$.

{\bf Step $5$.} If $B$ meets $2$ then it measures the qubit $|\psi\rangle=a|0\rangle+b|1\rangle$. If $f(X^2)=PartialMOD_{m_2}^\beta(X^2)=1$, then the qubit is rotated by an angle $\pi/2+v\cdot \pi$, for some integer $v$, else the qubit is rotated by an angle $w\cdot \pi$, for some integer $w$. Note that  before Step 4 if $y_2=1$, then $|\psi\rangle=|1\rangle$; and if $y_2=0$ then $|\psi\rangle=|0\rangle$. Therefore, if $y_3=PartialMOD_{m_2}^\beta(X^2)\oplus y_2 =1$, then $a\in\{1,-1\}$ and $b=0$. If $y_3=PartialMOD_{m_2}^\beta(X^2)\oplus y_2 =0$, then $a=0$ and $b\in\{1,-1\}$. The algorithm measures $|\psi\rangle$ and outputs $y_3 = PartialMOD_{m_2}^\beta(X^2)\oplus y_2$. 

{\bf Step $i$.} The step is similar to Step 4, but the algorithm reads $X^{i-1}$ and calculates $PartialMOD_{m_2}^\beta(X^{i-1})$.

{\bf Step $i+1$.} The step is similar to Step 5, but the algorithm outputs $y_{i} = PartialMOD_{m_2}^\beta(X^{i-1})\oplus y_{i-1}$.

{\bf Step $2k+2$.} The algorithm reads and skips the last part of the input. $B$ does not need these variables, because it guesses $y_1$ and using this value we already can obtain $y_{2},...,y_{k}$ without $X^{k}$.
\vspace{0.3cm}

Let us describe an algorithm $A$ with an advice bit for $BH^t_{k,r,w}(f)$, $f=PartialMOD_m^\beta$.

{\bf Step 1.} Algorithm $A$ gets $g_1=\bigoplus\limits_{j=1}^{k}f(X^j)$ as the advice bit, outputs it as $y_1$ and initialize $|\psi\rangle=|g_1\rangle$.

Other steps are similar to the algorithm $B$.
\Endproof

This theorem gives us the following important results.
%
Firstly, there is a quantum online streaming algorithm with $1$ qubit of memory for  $BHM^t_{k,r,w}$ having a better competitive ratio than
 any classical (deterministic or randomized) online streaming algorithm with sublogarithmic size of  memory, even if the classical one uses advice bits.
%
Secondly, there is an optimal quantum online streaming algorithm with $1$ qubit of memory and $1$ advice bit for the problem.
%
Finally, in a case of the IBHM Problem, a quantum online streaming algorithm with a constant size of memory and a constant number of advice bits is better than any deterministic online algorithm with unlimited memory and the same number of advice bits.

%
%
\bibliographystyle{alpha}
\bibliography{tcs}
\appendix
\section{Definition of OBDDs}\label{apx:obdd}

OBDD is a restricted version of a branching program (BP). BP over a set $X$ of $n$ Boolean variables is a directed acyclic graph with two distinguished nodes $s$ (a source node) and $t$ (a sink node). We denote it $P_{s,t}$ or just $P$. Each inner node $v$ of $P$ is associated with a variable $x\in X$. A {\em deterministic} BP has exactly two outgoing edges labeled $x=0$ and $x=1$ respectively for each node $v$. The program $P$ computes a Boolean function $f(X)$ ($f:\{0,1\}^n \rightarrow \{0,1\}$) as follows: for each $\sigma\in\{0,1\}^n$ we let $f(\sigma)=1$ iff there exists at least one $s-t$ path (called {\em accepting} path for $\sigma$) such that all edges along this path are consistent with $\sigma$. A {\em size} of branching program $P$ is a number of nodes.
Ordered  Binary Decision Diagram (OBDD) is a BP with following restrictions: 
(i) Nodes can be partitioned into levels $V_1, \ldots, V_{\ell+1}$ such that  $s$ belongs to the first level  $V_1$ and sink node $t$ belongs to the last level $V_{\ell+1}$. Nodes from level $V_j$ have outgoing edges only to nodes of level $V_{j+1}$, for $j \le \ell$.
(ii)All inner nodes of one level are labeled by the same variable.
(iii)Each variable is tested on each path only once.
A {\em width} of a program $P$ is $ width(P)=\max_{1\le j\le \ell}|V_j|. $ 
OBDD $P$ reads variables in its individual  order
$\theta(P)=(j_1,\dots,j_n)$. Let $id=(1,\dots,n)$ be a natural order of input variables. If OBDD reads input variables in the order $id$, then we denote the model as id-OBDD.

Probabilistic OBDD (POBDD) can have more than two edges for a node, and we choose one of them using a probabilistic mechanism. POBDD $P$ computes a Boolean function $f$ with bounded error $0.5-\varepsilon$ if probability of the right answer is at least $0.5+\varepsilon$.

Let us define a quantum OBDD. It is given in different terms, but they are equivalent, see \cite{agkmp2005} for more details.
For a given $ n>0 $, a quantum OBDD $ P$ of width $d$ defined on $ \{0,1\}^n $, is a 4-tuple
$
    P=(T,|\psi\rangle_0,Accept,\pi),
$
where
$ T = \{ T_j : 1 \leq j \leq n \mbox{ and } T_j = (G_j^0,G_j^1) \} $ are ordered
pairs of (left) unitary matrices representing transitions. Here $ G_j^0 $ or
$ G_j^1 $ is applied on the $j$-th step. A choice is determined by the input bit.  
The vector $|\psi\rangle_0$ is a initial vector from the $ d $-dimensional Hilbert space over the field of complex numbers.  $ |\psi\rangle_0=|q_0\rangle$ where $ q_0 $ corresponds to the initial node.
$ Accept \subset \{1,\ldots,d\} $ is a set of accepting nodes.
$ \pi $ is a permutation of $ \{1,\ldots,n\} $. It defines the order of input bits.
%
  For any given input $ \nu \in \{0,1\}^n $, the computation of $P$ on $\nu$ can be traced by the $d$-dimensional vector from a Hilbert space over the field of complex numbers. The initial one is $ |\psi\rangle_0$. In each step $j$, $1 \leq j \leq n$, the input bit $ x_{\pi(j)} $ is tested and then the corresponding unitary operator is applied:  
$
    |\psi\rangle_j = G_j^{x_{\pi(j)}} (|\psi\rangle_{j-1}),
$ 
where  $ |\psi\rangle_j $ represents the state of the system after the $ j$-th step.
The quantum OBDD can measure one or more qubits on any steps. Let the program be in state $|\psi\rangle=(v_1,
\dots, v_d)$ before a measurement and let us measure the $i$-th qubit.
Let states with numbers $j^0_1,\dots,j^0_{d/2}$ correspond to the $0$ value of the $i$-th qubit, and states with numbers $j^1_1,\dots,j^1_{d/2}$ correspond to the $1$ value of the $i$-th qubit.
The result of the measurement of the $i$-th qubit is $1$ with probability $pr_1= \sum_{z=1}^{d/2}|v_{j^1_z}|^2$ and $0$ with probability $pr_0=1-pr_1$.
The program $P$ measures all qubits at the end of the computation process. The program  accepts the input $ \sigma $ (returns $1$ on the input) with probability $
    Pr_{accept}(\nu)=\sum_{i \in Accept} v^2_i
$, for $ |\psi\rangle_n=(v_1,\dots,v_d)$.
$P_{\varepsilon}(\nu)=1$ if $P$ accepts input $\nu\in\{0,1\}^n$ with  probability at least $ 0.5+\varepsilon$, and $P_{\varepsilon}(\nu)=0$ if  $P$ accepts the input $\nu\in\{0,1\}^n$ with  probability at most $ 0.5-\varepsilon$, for $ \varepsilon \in (0,0.5] $.
We say that a function $f$ is computed by $ P$ with a bounded error if there exists $ \varepsilon \in (0,0.5] $ such that $ P_{\varepsilon}(\nu)=f(\nu)$ for any  $\nu\in\{0,1\}^n$. We can say that $P$ computes $f$ with a bounded error $0.5-\varepsilon$. 

We can say that an automaton is an id-OBDD such that a transition function for each level is the same. Note that id-OBDD is OBDD with an order $id=(1,\dots,n)$.

\end{document}